\documentclass[11pt]{article}
\usepackage{enumitem}
\usepackage{graphicx}
\usepackage{fancybox}
\usepackage{comment}
\usepackage{xcolor}
\usepackage{nameref}
\usepackage{bbm}

\definecolor{ForestGreen}{rgb}{0.1333,0.5451,0.1333}
\definecolor{DarkRed}{rgb}{0.8,0,0}
\definecolor{Red}{rgb}{1,0,0}

\usepackage{amsmath}

\usepackage{amssymb}
\usepackage{amsthm}

\usepackage{subfig}

\usepackage{thm-restate}

\usepackage{float}

\def\setof#1{\left\{#1  \right\}}

\newcommand\R{\mathbb{R}}

\newcommand{\E}[1]{\mathop{{}\mathbb{E}}\left[#1\right]}

\newcommand{\polylog}[1]{\mathrm{polylog}(#1)}

\newcommand{\dist}{\operatorname{dist}}

\renewcommand{\hat}{\widehat}
\renewcommand{\tilde}{\widetilde}

\renewcommand{\epsilon}{\ensuremath\varepsilon}

\renewcommand{\phi}{\ensuremath{\varphi}}

\newtheorem{theorem}{Theorem}[section]

\newtheorem{lemma}[theorem]{Lemma}

\newtheorem{claim}[theorem]{Claim}

\newtheorem{definition}[theorem]{Definition}
\newtheorem{remark}[theorem]{Remark}

\newtheorem*{theorem*}{Theorem}

\newtheorem*{corollary*}{Corollary}
\newtheorem*{conjecture*}{Conjecture}
\newtheorem*{lemma*}{Lemma}
\newtheorem*{thm*}{Theorem}
\newtheorem*{prop*}{Proposition}
\newtheorem*{obs*}{Observation}
\newtheorem*{definition*}{Definition}

\newtheorem*{remark*}{Remark}
\newtheorem*{rec*}{Recommendation}

\usepackage{fullpage}
\usepackage{algorithm}
\usepackage{algorithmic}
\usepackage{amsmath}
\usepackage{graphicx}
\usepackage{float}
\usepackage[colorinlistoftodos]{todonotes}

\usepackage{hyperref}
\usepackage{cleveref}
\providecommand{\Description}[1]{}
\crefname{figure}{figure}{figures}
\Crefname{figure}{Figure}{Figures}

\author{
    Yves Baumann \\
      \small ETH Zurich \\
  \small \texttt{ybaumann@ethz.ch}
  \and
  Gernot Z\"ocklein\thanks{The research leading to these results has received funding from the starting grant “A New Paradigm for Flow and Cut Algorithms” (no. TMSGI2 218022) and grant no. 200021 204787 of the Swiss National Science Foundation.} \\
  \small ETH Zurich \\
  \small \texttt{gernot.zoecklein@inf.ethz.ch}
}

\date{}
\begin{document}
\title{Parallel Spectral Graph Sparsification via Low Diameter Decompositions}

\maketitle
\begin{abstract}
    We present a new solver-free parallel spectral sparsification algorithm for weighted graphs that relies only on parallel low-diameter decompositions and independent sampling. This yields the first algorithmic improvement over prior, solver-free parallel sparsification approaches since Koutis (2014) and, for the first time for a practical algorithm, eliminates any dependence on the target approximation accuracy $\epsilon$ in the algorithm’s work and depth.
    
    Our algorithm works by sub-sampling edges according to their robust connectivity, as introduced by Kapralov and Panigrahy (2012). We show how to estimate the robust connectivities of $G$ in an extremely simple manner: we create multiple random sub graphs $G_p$, where each edge in $G$ is sub-sampled independently with probability $p_e = \min \setof{w_e \cdot p, 1}$. Then, we run a Low Diameter Decomposition in each of the graphs. If $u$ and $v$ often share a cluster in the LDDs, then this provides us with an upper bound on the robust connectivity of the edge $e = (u,v)$. Carefully invoking this procedure for $O(\log n)$ different values of the probabilities $p$ then allows us to obtain sufficiently good estimates for sub-sampling.
    
    We additionally complement the theory with an experimental evaluation demonstrating strong performance across relevant graphs and sparsity regimes. 
\end{abstract}

\section{Introduction}
Among the most important computational kernels in scientific computing is solving a linear system
\begin{equation*}
    Ax = b.
\end{equation*}
In many applications, the matrix $A$ is large, sparse, and symmetric positive semidefinite, so solving this system via a direct factorization is often impractical. Instead, practitioners resort to iterative methods such as Conjugate Gradient (CG), which repeatedly refines an initial guess and only requires matrix--vector product access to $A$.

Iterative methods also allow for \emph{preconditioning}. If $P$ is full rank (or shares the same kernel as $A$), then the systems
\begin{equation*}
    P^{-1}Ax = P^{-1}b
    \qquad\text{and}\qquad
    Ax = b
\end{equation*}
have the same solution. If the preconditioned system converges faster, it is advantageous to solve it instead. For CG, convergence to error $\delta$ is bounded by the condition number $\kappa$ of the preconditioned system \cite{saad2003}. For an initial guess $x_0$ and the exact solution $x_*$, we have the bound:
\begin{equation*}
    \|x_* - x_k\|_{A}
\;\le\;
2\left(\frac{\sqrt{\kappa}-1}{\sqrt{\kappa}+1}\right)^k
\|x_* - x_0\|_{A},
\qquad
\textnormal{where } \|v\|_{A} := \sqrt{v^\top A v},
\end{equation*}
\begin{equation*}
\kappa := \frac{\lambda_{\max}(A)}{\lambda_{\min}(A)}.
\end{equation*}
This dependence on $\kappa$ motivates the design of preconditioners that both (i) reduce the condition number substantially and (ii) can be applied efficiently in each CG iteration.

There has been extensive work on constructing good preconditioners. Incomplete Cholesky preconditioners are effective for some general sparse SPSD matrices, while multigrid preconditioners are highly successful for structured systems arising from numerical PDE discretizations. However, for graph-structured problems and other irregular sparsity patterns, designing parallelizable preconditioners with provable guarantees remains challenging.

In their seminal paper, Spielman and Teng \cite{spielman2008nearlylineartimealgorithmsgraph} showed that a linear system with a diagonally dominant matrix can be solved in nearly linear time in the number of nonzeros, via sparse \emph{spectral} approximations that lead to fast convergence of iterative methods.

Diagonally dominant matrices can be interpreted as graph Laplacians. Let $G=(V,E,w)$ be a weighted graph with Laplacian $L$. We say that a re-weighted sub graph $\tilde G$ of $G$, with Laplacian $\tilde L$, is a $(1 \pm \epsilon)$-spectral sparsifier  of $G$ if
\begin{equation*}
    (1- \epsilon)\,x^\top Lx \le x^\top \tilde Lx \le (1+\epsilon)\,x^\top Lx \quad   \forall x\in\mathbb{R}^V.
\end{equation*}
While graphs can be very dense ($|E| \approx |V|^2$), it is known \cite{twiceramanujan} that any graph admits a $(1 \pm \epsilon)$-spectral sparsifier with at most $\mathcal{O}(|V|/\epsilon^2)$ edges. Thus, one can replace a dense graph by a much sparser one while approximately preserving its spectral structure, making downstream linear-algebraic primitives cheaper.

A key application is solving Laplacian linear systems $Lx=b$. Such systems arise not only in scientific computing, but also as subroutines in graph optimization (e.g., max-flow and min-cost flow), in machine learning, and increasingly in modern graph-based pipelines including graph neural networks. Now, note that if $\Tilde{L}$ is a spectral sparsifier of $L$, this in particular implies that the condition number of $\Tilde{L}^{-1} L$ is $1 + O(\epsilon)$, so by using CG, we can solve linear systems in $L$ by solving roughly $O(\sqrt{\epsilon})$ systems in $\Tilde{L}^{-1}$. As approximate Gaussian elimination yields a factorization with $\mathcal{O}(|E|\,\polylog{|V|})$ nonzeros that can be used to solve Laplacian systems \cite{kyng2016approximategaussianeliminationlaplacians}, we profit immensely if the graph $\Tilde{G}$ is much sparser than $G$. Furthermore, repeated sparsification during factorization can yield polylogarithmic depth parallel constructions \cite{BK24}, making fast parallel sparsification a useful primitive both on its own and as a component inside parallel Laplacian solvers. For these reasons, we believe it to be an important problem to understand highly efficient and practical parallel constructions of sparse spectral sparsifiers. 

\paragraph{Our Contribution}
We revisit spectral sparsification through the lens of robust connectivity, as introduced by \cite{Kapralov2014}. Their approach incurs substantial polylogarithmic overhead in both depth and work due to the reliance on approximate distance oracles. We demonstrate that these overheads can be eliminated entirely when estimating robust connectivities.

Our first insight is that it is unnecessary to compute the exact distances $\dist(u, v)$ in the relevant auxiliary graphs. Instead, it suffices to detect the cases $\dist(u, v) = O(\log n)$ and $\dist(u, v) = \Omega(\log^2 n)$. This is a significantly simpler task than approximating arbitrary distances.

Our second insight is that, because robust connectivity is estimated by counting the occurrence of an event on random sub graphs, correctness in expectation is sufficient. Concretely, it is enough to sample random variables $X_{u,v}$ such that $\Pr[X_{u,v} = 1] \ge 1/2$ whenever $\dist(u, v) \le O(\log n)$, and $X_{u,v} = 0$ deterministically whenever $\dist(u, v) \ge \Omega(\log^2 n)$. Notice that this avoids an extra $O(\log n)$ factor in work that would be required for a high probability statement.

We show that such random variables can be obtained via a single low-diameter decomposition (LDD): we set $X_{u,v} = 1$ if and only if $u$ and $v$ lie in the same cluster. Importantly, low-diameter decompositions admit highly efficient parallel implementations, both in theory and in practice, yielding substantial improvements over prior work.

Then, we show how we can use the estimated robust connectivities to construct practical, sparse preconditioners for Laplacian linear systems. We show that our proposed sparsification algorithm outperforms baselines by a factor of $4\times$ to $5\times$ on the largest graphs, measured by the number of PCG iterations used until convergence. And that it is robust with respect to the desired sparsity of the preconditioner.

Our parallel algorithm, to the best of the authors' knowledge, is the first solver-free, practical spectral sparsification algorithm with no dependency on the desired accuracy $\epsilon$ in the depth and work. As such, it is highly relevant as a subroutine in parallel approximate Cholesky algorithms, which need to repeatedly compute sparsifiers with accuracy $\frac{\epsilon}{\log n}$. We outline how using our algorithm as a black box in the known parallelization framework from \cite{BK24} leads to better work and depth guarantees to solve Laplacian linear systems.

\paragraph{Related Work}
Our algorithm builds heavily on the idea of robust connectivities, as introduced in \cite{Kapralov2014}. Unfortunately, as the authors note, their algorithm requires access to an approximate distance oracle, and at the time of publication, no satisfactory parallel distance oracle was available. Very recently, in \cite{kyng2025random}, a $O(\log n)$ approximate distance oracle with $O(m\,\polylog{n})$ work and $O(\polylog{n})$ depth was discovered. Using this distance oracle inside \cite{Kapralov2014} results in a combinatorial algorithm that constructs a spectral sparsifier of $G$ with $O(n \epsilon^{-2} \log^3 n)$ edges. This algorithm requires $O(m \,\polylog{n})$ work and $O(\polylog{n})$ depth. Unfortunately, the polylog factors in \cite{kyng2025random} are quite large, mostly due to repeated calls to parallel approximate SSSP algorithms, and we are not aware of a practical parallel implementation.

In \cite{koutis2014}, Koutis introduced an alternative combinatorial parallel algorithm based on graph spanners, and the analysis of Koutis' algorithm was later refined in \cite{soarasmus}. Their algorithm produces sparsifiers with $O(n \epsilon^{-2}\log^2 n \log \log n)$ many edges, its depth is \\$O(\epsilon^{-2} \log^4 n \log^* n)$, and its work $O(m \epsilon^{-2} \log^2 n)$. Note here the factor of $\epsilon^{-2}$ in both its depth and work. For values $\epsilon = O(1 / \log n)$, this factor causes a significant overhead. Note also that our algorithm improves over the depth of the algorithm by a factor of $O(\log^2 n)$, although our sparsity is larger by a factor of $O(\log n)$.

In \cite{doi:koutis2015}, it was shown how to find $2$-approximations to the effective resistances of a graph $G$ by essentially solving $O(\log n)$ Laplacian linear systems. Consequently, one can use the parallel spectral sparsification algorithm from \cite{10.1145/3558481.3591101} to get a spectral sparsifier of $G$ with $O(n \epsilon^{-2} \log n)$ many edges in $O(\log^2 n \log \log n)$ depth and $O(m \log^4 n \log \log n)$ or $O((m+n \log^5 n)\log^2 n \log \log n)$ work. Consequently, for graphs with more than $\Omega(n\log^5 n)$ edges, this result matches our result both in work and depth, although it yields a graph that is sparser by a factor of $O(\log^2 n)$. We remark that one of the main practical applications of spectral sparsification is inside Laplacian Solvers, and so the practical relevance of a solver-based sparsification algorithm is unclear. We also believe that even from a theoretical viewpoint, combinatorial and solver-free spectral-sparsification algorithms are insightful and interesting in their own right.

%\gernot{yves, can you write something here about the practical part}

\section{Preliminaries}
\paragraph{Notation}
In this paper, we consider weighted undirected graphs  $G = (V, E, w)$. Throughout, we denote by $n$ the number of vertices of $G$, $m$ the number of edges, and $W = \max_{e \in E} w_e$ is the maximum weight of $G$. We denote by $L_G$ the graph Laplacian of $G$. We denote by $\dist_G(u, v)$ the shortest path distance between $u$ and $v$ in graph $G$. We call a set $\mathcal{X}$ of subsets of $V$ a partition of $G$ if $\bigcup_{X \in \mathcal{X}} X = V$ and any two distinct sets $X, X' \in \mathcal{X}$ are disjoint: $X \cap X' = \emptyset$. Given a set $X \subseteq V$, we denote by $G[X]$ the induced subgraph on $X$. We write that $H \preceq G$ if for every $x \in \R^n$, $x^T L_H x \leq x^T L_G x$. We denote by $R_e$ the effective resistance of edge $e = (u, v)$ in $G$, i.e., $R_e = (\mathbf{1}_u - \mathbf{1}_v)^T L_G^{+} (\mathbf{1}_u - \mathbf{1}_v)$. We call the quantity $w_e R_e$ the leverage score of $e$.

For a sampling parameter $p \in [0,1]$ we denote by $G_p$ the \emph{unweighted} random graph obtained by sampling each edge $e \in E$ independently with probability $\min\{w_e p, 1\}$.

\paragraph{Parallel Model}
In this paper, we assume the CRCW PRAM model.

\subsection{Effective Resistance Sampling}
Our algorithm works by independent sampling of edges according to their \emph{effective resistances}, as first introduced in \cite{doi:er_sampling}.

\begin{theorem}[\cite{doi:er_sampling}]\label{thm:er_sampling}
 Let $H$ be obtained by sampling edges of $G$ independently with probability $p_e \geq c w_e R_e \log n / \epsilon^2$ for some $\epsilon>1 / \sqrt{n}$ and a sufficiently large constant $c>0$, and {rescaling the weight of each sampled edge by $1 / p_{\mathrm{e}}$}. Then with high probability,
$$
(1-\epsilon) G \preceq H \preceq(1+\epsilon) G .
$$
\end{theorem}
{We will also use the following basic variant of a Chernoff Bound:
\begin{theorem}\label{thm:chernoff}[See Exercise 4.7 in \cite{mitzenmacher2005probability}]
    Let $\setof{X_i}_{i=1}^k$ be a sequence of independent Bernoulli random variables. Let $X = \frac{1}{k} \sum_{i=1}^k X_i$, and $\mu := \E{X}$. Then for any $\delta \in [0, 1]$, we have that
    $$
    \Pr[X \leq (1- \delta) \mu] \leq e^{- \delta^2 k \mu / 2}.
    $$
    Moreover, for any $c \geq \mu$ and $\delta \in [0, 1]$, we additional have that
    $$
    \Pr[X \geq (1+\delta) c] \leq e^{- \delta^2 k c / 3}.
    $$
\end{theorem}
We need this slightly non-standard version because in the proof of \Cref{lma:separation} the only control we have over $\E{X}$ is that $\E{X} \leq 1/2$. But in the lemma, our goal is to show that with high probability $X < 8.5/16$. As $\E{X}$ may be much smaller than $1/2$, a standard-chernoff bound written only in terms of an exponential decay in terms of $\E{X}$ is thus not sufficient for our purposes.}

\subsection{Robust Connectivities}

To efficiently estimate the leverage scores of the edges of $G$, we make use of the notion of robust connectivities, introduced by Kapralov and Panigrahy \cite{Kapralov2014}.

\begin{definition}
Let $G = (V, E, w)$ be a weighted undirected graph. For $\eta \in [0,1]$, $\kappa \geq 1$ and {$e =(u, v) \in E$ we define 
\[
p_\kappa(e, \eta) := \Pr[\dist_{G_\eta}(u, v) > \kappa].
\]}
\end{definition}

We stress here again that while $G$ is a weighted graph, the graph $G_{\eta}$ is unweighted. We also remark that clearly $p_{\kappa}(e, \eta)$ is a decreasing function with increasing $\eta$, i.e., that $p_{\kappa}(e, \eta) \leq p_{\kappa}(e, \eta')$ for $\eta \geq \eta'$.

\begin{definition}[Robust Connectivity]\label{def:robust_connectivities}
For $e = (u, v) \in E$ we let the $\kappa$-robust connectivity $q_\kappa(e)$ 
denote the largest {$\eta \in [0, 1]$ such that $p_\kappa(e, \eta) = \Pr[\dist_{G_\eta}(u, v) > \kappa] \ge 1/2$.}
\end{definition}
{
The next lemma shows that sampling according to scaled versions of the robust connectivities constitutes a valid leverage-score based over-sampling.}

\begin{lemma}[Lemma $2$ in \cite{Kapralov2014}]\label{lem:robust_connectivities_bound}
For all edges $e \in E$,
\[
R_e \leq 2 \kappa \cdot q_\kappa(e).
\]
\end{lemma}
Kapralov and Panigrahy \cite{Kapralov2014} also proved that the weighted sum of robust connectivities is bounded from above:
\begin{lemma}[Lemma $3$ in \cite{Kapralov2014}]\label{lem:robust_connectivities_bound2}For a weighted undirected graph $G = (V, E, w)$, we have that
\[
\sum_{e \in E} w_e q_\kappa(e) \;\le\; 2n^{\,1 + O(1/\kappa)}.
\]
\end{lemma}
{As we will later binary-search for the values of the robust connectivities, we also need to show that the minimum of the robust connectivities is suitably lower bounded.}
\begin{claim}\label{clm:lb_qk}
    For a graph $G = (V, E, w)$ with maximum weight $W$, and any $\kappa > 1$, we have that 
    $$\min_{e \in E} q_{\kappa}(e) \geq \frac{1}{4 \cdot |E| \cdot W}.$$
\end{claim}
\begin{proof}
    We show that for $p = 1 / (4 |E| W)$, any pair $u, v \in V$ satisfies $\Pr[u, v \textnormal{ disconnected in } G_p] > 1/2$, from which the claim then immediately follows from the definition of $q_{\kappa}(e)$. To do so, consider any fixed cut $(S, V \setminus S)$ with $u \in S$ and $v \in V \setminus S$. Let $E(S, V \setminus S)$ denote the edges with one endpoint in $S$ and the other in $V \setminus S$. Then 
    $$
    \E{|E_{G_p}(S, V \setminus S)|} = \sum_{e \in E_G(S, V \setminus S)} \frac{w_e}{4 |E| W} \leq 1/4.
    $$
    Consequently, by Markov's inequality, we must have
    $$
    \Pr[u, v \textnormal{ connected in } G_p] \leq \Pr[|E_{G_p}(S, V \setminus S)| \geq 1] \leq 1/4.
    $$
\end{proof}
We finally need the following auxiliary claim, which will later be crucial for the analysis of our algorithm.
\begin{claim}\label{clm:help}
    {Let $\pi \in (0, 1]$. Then} for an edge $e$ with $q_{\kappa}(e) \leq \pi / 4$, we have that $p_{\kappa}(e, \pi) \leq 1/4 $.
\end{claim}
\begin{proof}
{We prove the claim by comparing \(G_\pi\) to the union of two independent copies of \(G_{\pi/4}\). The proof has two steps. First, the assumption \(q_\kappa(e) \le \pi/4\) implies that a single copy of \(G_{\pi/4}\) fails to contain a \(u\)-to-\(v\) path of length at most \(\kappa\) with probability at most \(1/2\). Second, the union of two independent copies of \(G_{\pi/4}\) is stochastically dominated by \(G_\pi\). Therefore, if \(G_\pi\) has no \(u\)-to-\(v\) path of length $\kappa$, then the union of the two smaller samples also has no such path, and this can happen only if both independent copies fail to contain such a path.}

Let us formally start by noting that by the definition of $q_{\kappa}(e)$ as the largest $\eta$ with $p_{\kappa}(e,\eta)\ge 1/2$, the assumption $q_{\kappa}(e)\le \pi/4$ implies
\[
p_{\kappa}\big(e,\tfrac{\pi}{4}\big)\le \tfrac12.
\]

Now take two independent copies of $G_{\pi/4}$ (call them $G^{(1)}_{\pi/4}$ and $G^{(2)}_{\pi/4}$) and let $H$ be the graph whose edge set is the union of the edge sets of these two copies. Consider any fixed edge $f$. Its inclusion probability in a single copy $G_{\pi/4}$ is
\[
p'_f=\Pr[f\in G_{\pi/4}]=\min\{w_f\cdot (\pi/4),\,1\}.
\]
Hence the inclusion probability of $f$ in $H$ is
\[
\Pr[f\in H]=1-(1-p'_f)^2 = 2p'_f - (p'_f)^2 \le 2p'_f.
\]
If $p'_f<1$, then $2p'_f \le 2\cdot(w_f \pi/4) = w_f \pi/2 \le w_f \pi$, and if $p'_f=1$, then $\Pr[f\in H]=1\le 1=\min\{w_f \pi,1\}$. Thus in all cases
\[
\Pr[f\in H] \le \min\{w_f\pi,\,1\} = \Pr[f\in G_\pi].
\]
{Therefore, there is an edgewise monotone coupling under which $H\subseteq G_\pi$ almost surely, i.e., $E_H$ is stochastically dominated by $E_{G_{\pi}}$.}

Because $H\subseteq G_\pi$, the event $\{\dist_{G_\pi}(u,v)>\kappa\}$ (no $u\!-\!v$ path of length $\le\kappa$ in $G_\pi$) implies $\{\dist_H(u,v)>\kappa\}$. Hence
\[
p_{\kappa}(e,\pi)=\Pr[\dist_{G_\pi}(u,v)>\kappa] \le \Pr[\dist_H(u,v)>\kappa].
\]

But $H$ is the union of two independent copies of $G_{\pi/4}$, so $H$ can have no $u\!-\!v$ path of length $\le\kappa$ only if \emph{both} copies have no such path. Therefore
\[
\Pr[\dist_H(u,v)>\kappa] \leq \big(p_{\kappa}(e,\tfrac{\pi}{4})\big)^2 \le \big(\tfrac12\big)^2 = \tfrac14.
\]

Combining the inequalities gives $p_{\kappa}(e,\pi)\le 1/4$, as required.
\end{proof}

{Note that estimating robust connectivities can be done by answering approximate distance queries between vertices in a graph. While sequential constructions such as that of Thorup and Zwick~\cite{Thorup_Zwick} produce efficient approximate distance oracles with provable stretch-space tradeoffs, parallel constructions with similar guarantees are more complex. In the following sections, we will show how one can compute these estimates without resorting to approximate distance queries and instead use low-diameter decompositions, which are easy and fast to compute in parallel.}

\subsection{Low Diameter Decompositions}
We need the following tool. Given a partition $\mathcal{X}$ of the vertex set of an unweighted graph $G = (V, E)$, define by $E_{del}$ the set of edges of $G$ with endpoints in different clusters.
\begin{definition}[$(\beta, D)$-LDD]
    We call a random partition $\mathcal{X}$ of $V$ a $(\beta, D)$-LDD, if, 
    \begin{enumerate}
        \item For every $u, v \in X \in \mathcal{X}$, we have $\dist_{G[X]}(u, v) \leq D$, and
        \item for every $e \in E$, we have that $\Pr[e \in E_{del}] \leq \beta$.
    \end{enumerate}
\end{definition}
Note the following consequence.
\begin{claim}\label{clm:ldd-path}
    If $\mathcal{X}$ is a $(\beta, D)$-LDD, and $u, v \in V$ are vertices satisfying $\dist_G(u, v) \leq \frac{1}{4 \beta}$, then
    $$
    \Pr[u, v \in X \in \mathcal{X}] \geq 3/4.
    $$
\end{claim}
\begin{proof}
    Let $P_{u, v} = e_1, \ldots, e_l$ be a $u$-$v$ shortest path of length $l \leq \frac{1}{4 \beta}$. By a union bound, 
    \begin{align*}
            \Pr[u, v \textnormal{ in different clusters}] &\leq \Pr[\exists i \in [l]: e_i \in E_{del}]\\
            &\leq \sum_{i=1}^l \beta = l \beta \leq 1/4. 
    \end{align*}
\end{proof}

\begin{theorem}[See \cite{mpxspanner} and \cite{mpx}]\label{thm:ldd}
    There is an algorithm that takes an unweighted graph with $n$ vertices, $m$ edges, a parameter $\beta \leq 1 / 2$ and samples from a $\left( \beta, \alpha_{LDD} \cdot \beta^{-1} \right)$-LDD for some $\alpha_{LDD} \in O\left(\log n\right)$. The algorithm requires $O\left(\beta^{-1} \log n \log^* n\right)$ depth and $O(m)$ work.
\end{theorem}

We remark that a highly efficient implementation of the above algorithm exists \cite{dhulipala2018theoretically}.

\section{$R_e$ estimation.}
We start by presenting the main lemma that our algorithm is based on. Essentially, given a parameter $\pi$, we show how to compute whether or not the robust connectivity of an edge satisfies $q_{\kappa}(e) \leq O(\pi)$, up to an error of $O(\alpha_{LDD})$. This will then allow us to binary search for good approximations of the robust connectivities. See \Cref{alg:ldd_edge_filter} for pseudocode. {Intuitively, \cref{alg:ldd_edge_filter} creates sparser copies of the initial graph and uses low-diameter decompositions as an indicator of whether the endpoints of an edge are still close together. If the endpoints are too far apart often, then the algorithm takes that as evidence that the leverage score is higher than the current parameter checked.}

\begin{algorithm}[H]
\caption{Robust Connectivity Decide}
\label{alg:ldd_edge_filter}
\begin{algorithmic}[1]

\REQUIRE Graph $G=(V,E)$; parameters $\kappa \geq 1$, $\pi \in [0,1]$
\STATE $\beta \gets 1/(4\kappa)$
\STATE $t \gets c \log n$ for sufficiently large constant $c$

\FOR{$i = 1$ to $t$ \textbf{in parallel}}
    \STATE Sample an independent copy $G_{\pi}^{(i)} \sim G_{\pi}$
    \STATE Run $\mathrm{LDD}(G_{\pi}^{(i)}, \beta)$
    \STATE Let $C^{(i)}(v)$ denote the component ID of each vertex $v \in V$
\ENDFOR

\FOR{each edge $e=(u,v) \in E$ \textbf{in parallel}}
    \STATE $Z_e \gets \frac{1}{t} \sum_{i=1}^t \mathbf{1}\setof{C^{(i)}(u) = C^{(i)}(v)}$
\ENDFOR

\STATE $\widetilde{E} \gets \{\, e \in E \mid Z_e \ge 8.5/16 \,\}$

\STATE \textbf{return} $\widetilde{E}$

\end{algorithmic}
\end{algorithm}

\begin{lemma}\label{lma:separation}
    \Cref{alg:ldd_edge_filter}, given as input a graph $G = (V, E, w)$ and parameters $\kappa \geq 1, \pi \in [0, 1]$ returns an edge set $\Tilde{E} \subseteq E$ with the following properties.
    \begin{enumerate}
        \item For every edge $e$ with $q_{4 \cdot \kappa \cdot \alpha_{LDD}}(e) \geq \pi$, we have that $e \not \in \Tilde{E}$. \label{lma:separation:1}
        \item For every edge $e$ with $q_{\kappa}(e) \leq \pi / 4$, we have that $e \in \Tilde{E}$. \label{lma:separation:2}
    \end{enumerate}
    The algorithm requires $O(\kappa \log n \log^* n)$ depth, $O(m \log n)$ work and is correct with high probability.
\end{lemma}
\begin{proof}
    {The proof separates edges into two regimes according to their
    robust connectivity. The statistic $Z_e$ estimates the probability
    that the endpoints $(u, v)$ of $e$ are placed in the same LDD cluster after
    sampling $G_{\pi}$. We show that this probability is at least $9/16$
    for edges with $q_\kappa(e)\le \pi/4$, and at most $1/2$ for edges
    with $q_{4\kappa\alpha_{\mathrm{LDD}}}(e)\ge \pi$. Since the threshold
    $8.5/16$ lies between these two values, averaging
    $O(\log n)$ independent trials and applying a Chernoff bound
    separates the two cases for all edges simultaneously with high
    probability.

   % The algorithm samples in parallel $t$ independent copies of $G_{\pi}$, where $t$ is to be specified later. Call these graphs $G_{\pi}^{(1)}, \ldots, G_{\pi}^{(t)}$. It then runs an LDD from \Cref{thm:ldd} with parameter $\beta = 1/(4 \kappa)$ on each of these graphs.
    
    Formally, for every $i \in [t]$ and edge $e =(u, v) \in E$, let $X_e^{(i)}$ be the indicator variable that is $1$ if $u, v$ share an LDD component in $G_{\pi}^{(i)}$, and $0$ otherwise. That is, $X_e^{(i)} = \mathbf{1}\setof{C^{(i)}(u) = C^{(i)}(v)}$ in the notation of the pseudo-code. 
    
    We now first analyze the case \(q_\kappa(e)\le \pi/4\). We start by noticing that $\Pr[X_e = 1 | \dist_{G_{\pi}}(u, v) \leq \kappa] \geq 3/4$ by \Cref{clm:ldd-path}. So by writing $$\E{X_e} \geq \E{X_e|\dist_{G_{\pi}}(u, v) \leq \kappa} \cdot \Pr[\dist_{G_{\pi}}(u, v) \leq \kappa],$$ we can deduce that $\E{X_e} \geq \frac{3}{4} \Pr[\dist_{G_{\pi}}(u, v) \leq \kappa] = \frac{3}{4} (1 - p_{\kappa}(e, \pi))$. Observe now that by \Cref{clm:help}, for an edge $e$ with $q_{\kappa}(e) \leq \pi /4$, we also have $p_{\kappa}(e, \pi) \leq 1/4$. Consequently, for an edge with $q_{\kappa}(e) \leq \pi / 4$, we must have $\E{X_e} \geq (3/4)^2 = 9/16$.
    
    Let us now consider the edges $e$ for which $q_{4 \kappa \alpha_{\operatorname{LDD}}}(e) \geq \pi$. First note that as the clusters of the LDD have diameter at most $\alpha_{\operatorname{LDD}} \beta^{-1}$, it must be the case that $$\Pr[X_e = 1 | \dist_{G_{\pi}}(u, v) > \alpha_{LDD} \cdot \beta^{-1}] = 0.$$ Similar to before, by decomposing by conditional expectations, we can thus conclude that 
    \begin{align*}
    & \E{X_e} \leq \Pr[\dist_{G_{\pi}}(u, v) \leq \alpha_{LDD} \cdot \beta^{-1}] = 1 - p_{\alpha_{LDD} \cdot \beta^{-1}}(e, \pi).
    \end{align*}
    
    Now recall the definition of $q_{\alpha_{LDD} \cdot \beta^{-1}}(e)$ as the largest $\eta \in [0, 1]$ such that $p_{\alpha_{LDD} \cdot \beta^{-1}}(e, \eta) \geq 1/2$. As in the case we are currently considering we assume that $q_{\alpha_{LDD} \cdot \beta^{-1}}(e) \geq \pi$, we can thus conclude that also $p_{\alpha_{LDD} \cdot \beta^{-1}}(e, \pi) \geq 1/2$. Plugging this into the equation displayed above shows that $q_{\alpha_{LDD} \cdot \beta^{-1}}(e) \geq \pi$ indeed implies that $\E{X_e} \leq 1/2$.

    Now, let us apply the Chernoff bound from \Cref{thm:chernoff}. First, by choosing $t \in O(\log n)$ large enough, it follows immediately from the lower bound supplied in that theorem that $\Pr[Z_e \leq (1-\frac{1}{32}) \E{Z_e}] \leq n^{-10}$. In particular, by a union bound, with probability at least $1-n^{-8}$, it must be the case that for all edges $e$ with $q_{\kappa}(e) \leq \pi/4$ it simultaneously holds that $$Z_e \geq (1-\frac{1}{32}) \E{Z_e} \geq (1-\frac{1}{32}) \frac{9}{16} > 8.5 / 16.$$ Similarly, as $\E{Z_e} \leq 1/2$ for all edges $e$ with $q_{4 \kappa \alpha_{\operatorname{LDD}}}(e) \geq \pi$, we can use the upper-bound of the theorem (and choosing $c = 1/2$ in \Cref{thm:chernoff}) to conclude in an analogous fashion that all such edges simultaneously satisfy that $$Z_e < (1+\frac{1}{32})\frac{1}{2} < 8.5 / 16$$ with probability at least $1-n^{-8}$. Thus, the returned set $$\Tilde{E} := \setof{e \in E_G: Z_e \geq 8.5/16}$$ indeed satisfies the requirements with probability at least $1-2n^{-8}$ by one final application of the union bound.

    Note that the depth is given by the depth of \Cref{thm:ldd}, so it is $O(\beta^{-1} \log n \log^* n) = O(\kappa \log n \log^* n)$. The total work is $O(m \log n)$.}
\end{proof}

The next lemma shows how to use the previous lemma to efficiently binary-search for approximations to the robust connectivities. 
\begin{lemma}\label{lma:rc_estimates}
    There exists a parallel algorithm that, given as input a graph $G = (V, E, w)$ and $\kappa > 0$, computes estimates $\hat{q}(e)$ satisfying $q_{4 \kappa \alpha_{LDD}}(e) \leq \hat{q}(e) \leq 8 \cdot q_{\kappa}(e)$.
    The algorithm requires $O(\kappa \log n \log^* n)$ depth, $O(m \log n \log nW)$ work and is correct with high probability.
\end{lemma}
\begin{proof}
    {The algorithm searches over geometrically decreasing values
    \(\pi_i=2^{-i}\) with $i = 0, 1, \ldots, \log(8mW)$. Let $\Tilde{E}^{(i)}$ denote the returned set of edges. For a fixed value \(\pi_i\), \Cref{lma:separation} tells us that \Cref{alg:ldd_edge_filter} separates two cases: edges with 
    \(q_\kappa(e) \leq \pi_i /4\) are included in \(\widetilde E^{(i)}\), while edges with \(q_{4\kappa\alpha_{\mathrm{LDD}}}(e) \geq \pi_i\) are excluded. This allows us to binary-search for the value of $q_{\kappa}(e)$ by setting it to be equal to $\pi_j$, where $j$ is the last time that edge $e$ is included in the set $\Tilde{E}^{(j)}$. In this proof, we will first separately argue for the lower- and upper bound of the $\hat{q}_{\kappa}(e)$, and then conclude by showing the resource bounds. 
    
    To start formally, let us define for any edge $e \in E_G$ the quantity
        $$j_e :=  \max \left\{j: \textnormal{s.t. }e \in \Tilde{E}^{(j)}\right\}.$$ The estimates our algorithm returns are then $$\hat{q}_{\kappa}(e) = \pi_{j_e} = 2^{-j_e}.$$
        
        Let us now formally argue why this algorithm is correct. We can assume by a union bound that all calls to \Cref{lma:separation} succeed simultaneously. We prove the lower and upper bounds on
\(\widehat q(e)\) for an arbitrary fixed edge \(e\).

So let us fix an edge $e$. We start by proving that $q_{4 \kappa \alpha_{LDD}}(e) \leq \hat{q}(e)$. To do so, note that as $e \in \Tilde{E}^{(j_e)}$, we know by \Cref{lma:separation:1} of the preceding lemma that $q_{4\kappa \alpha_{LDD}}(e) < \pi_{j_e} = 2^{-j_e} = \hat{q}(e)$. This already finishes the proof of the lower bound.

We now prove that $\hat{q}_{\kappa}(e) \leq 8 q_{\kappa}(e)$. We proceed by a proof by contradiction, so assume that $q_{\kappa}(e) \leq 2^{-j_e-3} = \hat{q}_{\kappa}(e) / 8$. Then we know by \Cref{lma:separation:2} of \Cref{lma:separation}, that we would have that $e \in \Tilde{E}^{(j_e+1)}$, because $\Tilde{E}^{(j_e+1)}$ contains all edges $e'$ satisfying $q_{\kappa}(e') \leq 2^{-(j_e+1)}/4 = 2^{-j_e-3}$. This is in contradiction to the definition of $j_e$ being the \emph{largest} value of $j$ such that $e \in \Tilde{E}^{(j)}$. We can conclude that $q_{\kappa}(e) > 2^{-j_e-3} = \hat{q}_{\kappa}(e) / 8$. Note that in the last part we assumed that there is an instance of \Cref{alg:ldd_edge_filter} with parameter $2^{-(j_e+1)}$, i.e., we assumed that $j_e+1 \leq \log(8 m W)$. This is the case, as by \Cref{clm:lb_qk}, we must have that $q_{\kappa}(e) \geq 1/(4 m W)$, so that $j_e + 1 \leq \log(4 m W) + 1 \leq \log(8mW)$. This also explains the reason why we need to choose $L = \log(8mW)$: because $1/(4 m W)$ is a lower bound for our binary-search range on the robust connectivities. This concludes the proof of correctness.

We now argue about the resource bounds. First, we observe that the depth of the algorithm is given by the depth of \Cref{lma:separation}, i.e., the depth is in $O(\kappa \log n \log^* n)$. Additionally, the work is dominated by the work of all $O(\log nW)$ calls to the same lemma, so it is in $O(m \log n \log nW)$. This concludes the proof.}
\end{proof}
Using these estimates, we can now also construct estimates of the effective resistances, as described in the lemma below. 
\begin{lemma}\label{thm:er_estimates}
    There exists a parallel algorithm that, given as input a graph $G = (V, E, w)$ with maximum weight $W$ computes effective resistance estimates $\hat{R}_e$ satisfying $\hat{R}_e \geq R_e$ and $\sum_e w_e \hat{R}_e \leq O(n \log^2 n)$. The algorithm requires $O(\log^2 n \log^* n)$ depth, $O(m \log n \log nW)$ work and is correct with high probability.
\end{lemma}
\begin{proof}
     We compute the estimates $\hat{q}_e$ from \Cref{lma:rc_estimates}. Remember that by \Cref{lem:robust_connectivities_bound}, $R_e \leq 8 \kappa \alpha_{LDD} q_{4 \kappa \alpha_{LDD}}(e)$. As $q_{4 \kappa \alpha_{LDD}}(e) \leq \hat{q}(e)$, by setting $\hat{R}_e = 8 \kappa \alpha_{LDD} \hat{q}(e)$, we have that $\hat{R}_e \geq R_e$. Moreover, by \Cref{lem:robust_connectivities_bound2}, it is the case that
     \begin{align*}
              \sum_{e \in E} w_e \hat{R}_e &= 8 \kappa \alpha_{LDD} \sum_{e \in E} w_e \hat{q}(e)\\
              &\leq 8 \kappa \alpha_{LDD} \sum_{e \in E} 8 w_e q_{\kappa}(e)\\
              &\leq O(\log n) \kappa \cdot 2 n^{1+O(1/\kappa)}.
     \end{align*}
     Setting $\kappa \in O(\log n)$ large enough shows the claim.
\end{proof}
We can now prove the main theorem. Pseudocode is given in \Cref{alg:spectral_sparsifier}. {Intuitively, \cref{alg:spectral_sparsifier} estimates robust connectivities by running \cref{alg:ldd_edge_filter} at geometrically decreasing sparsities, which then allows us to bound the robust connectivity estimates with a relatively small error.} The pseudo-code combines the algorithmic steps from \Cref{lma:rc_estimates}, \Cref{thm:er_estimates}, and \Cref{thm:main_algorithm}.
\begin{theorem}\label{thm:main_algorithm}
    There is a parallel algorithm that produces a $(1 \pm \epsilon)$ spectral sparsifier of a graph $G$ with $O(n \epsilon^{-2} \log^3 n)$ many edges in work $O(m \log n \log nW)$ and depth $O(\log^2 n \log^* n)$ and is correct with high probability.
\end{theorem}
\begin{proof}
    We sample each edge with probability $p_e = c w_e \hat{R}_e \log n / \epsilon^2$ and {rescale} its weight by $1/p_e$ if sampled, where the constant $c$ is from \Cref{thm:er_sampling}. By the theorem, $H$ is a $(1 \pm \epsilon)$ spectral sparsifier of $G$ with high probability. Furthermore, the expected number of edges is given by $\sum_e p_e = O(\epsilon^{-2}\log n) \sum w_e \hat{R}_e = O(n \epsilon^{-2}\log^3 n)$ by \Cref{thm:er_estimates}. A standard Chernoff bound may now be applied to turn this into a high probability guarantee. 
\end{proof}

\begin{algorithm}[H]
\caption{Parallel Spectral Sparsification}
\label{alg:spectral_sparsifier}
\begin{algorithmic}[1]

\REQUIRE Graph $G = (V, E, w)$; accuracy parameter $\epsilon > 0$  

\STATE Let $L \gets \lceil \log(8 m W) \rceil$ // As $\frac{1}{4 m W}$ lower bounds all $q_{\kappa}(e)$.
\STATE Let $\kappa \in O(\log n)$ be large enough; $\alpha_{LDD}$ be the constant from \Cref{thm:ldd}.

\FOR{$i = 0$ to $L$ \textbf{in parallel}}
    \STATE Run \Cref{alg:ldd_edge_filter} with $\pi_i = 2^{-i}$
    \STATE Let $\widetilde{E}^{(i)}$ be the edges returned
\ENDFOR

\FOR{each edge $e \in E$ \textbf{in parallel}} 
    \STATE $j_e \gets \max \{\, i \mid e \in \widetilde{E}^{(i)} \,\}$
    \STATE $\hat{q}(e) \gets 2^{-j_e}$
    \STATE $\hat{R}_e \gets 8 \kappa \alpha_{\mathrm{LDD}} \cdot \hat{q}(e)$
\ENDFOR

\STATE $H \gets (V, \emptyset)$
\FOR{each edge $e \in E$ \textbf{in parallel}}
    \STATE $p_e \gets \min \left(1, c \cdot w_e \hat{R}_e \frac{\log n}{\epsilon^2}\right)$ \hspace{3em}// Here $c$ is the constant from \Cref{thm:er_sampling}.
    \STATE Sample edge $e$ with probability $p_e$
    \IF{$e$ is sampled}
        \STATE Add $e$ to $H$ with weight $w_H(e) \gets {w_e / p_e}$
    \ENDIF
\ENDFOR

\STATE \textbf{return} $H$

\end{algorithmic}
\end{algorithm}

\section{Experimental}
We validate our proposed method to compute leverage score estimates on a range of graphs. 

\subsection{Problem setting}
Spectrally sparsifying graphs is a key subroutine when solving Laplacian linear systems. The sparsified version of a graph can be used as a preconditioner in iterative methods such as conjugate gradient. If the sparsified graph has significantly fewer nonzero entries than the original graph, while being spectrally similar, then it drastically reduces the time to compute and apply the preconditioner. Furthermore, in \cite{BK24} it was shown that repeated sparsification of the remaining graph during approximate Cholesky leads to a polylogarithmic depth for computing the preconditioner. In this spirit, we propose the following experiment on two types of graphs.\\

We set up a multidimensional grid (2d-grid and 3d-grid) and run Cholesky on the underlying matrix until it is dense. Then, we sparsify the dense subgraph and see how well the sparsified graph spectrally approximates the dense subgraph.\\

\subsection{Algorithms and Baselines}
In \cref{thm:main_algorithm}, we propose an algorithm that spectrally sparsifies a graph through repeated LDD computations. While our theoretical bound of $O\left(\frac{n\log^3n}{\epsilon^2}\right)$ is very sparse in theory, it is often not sparse enough {in practice}, since a large constant as well as the exponent of the logarithm may well exceed the desired density. This observation holds for many simple sparsification methods, including the simple parallel sparsifier by Koutis \cite{koutis2014} or even just using the oversampling lemma together with leverage score estimates \cite{doi:er_sampling}. In fact, it is the reason why it is often said that sampling by leverage scores does not produce effective sparsifiers in practice. We instead propose a practical sparsifier motivated by the sparsifier of Koutis \cite{koutis2014} and show how to leverage our robust connectivity estimates to improve spectral similarity. It should be seen as a heuristic to compute good spectral sparsifiers based on our robust connectivity estimates.\\

Our algorithm takes as input a weighted graph $G$ and a target sparsity budget of $f\cdot m$ edges. It first computes, for every edge, an estimate $\hat{R}_e$ of its effective resistance (robust connectivities), and sets a sampling score $s_e \gets w_e \hat{R}_e$. Next, it constructs a maximum weight spanning tree $T$ where we replace the edge-weights by the score $s_e$. Finally, it samples each non-tree edge $e \in E\setminus T$ independently with probability
\[
p_e=\min\!\left(1,\frac{(f\,m - n+1)\,s_e}{\sum_{e'\in E\setminus T} s_{e'}}\right),
\]
and, if selected, sets the weight to $w_e/p_e$ to keep the expected contribution unbiased. The resulting sparsifier is $H=(V,\,T\cup \widehat{E})$. See \cref{alg:practical_sparsification} for pseudocode.

\begin{algorithm}[H]
\caption{Leverage-Based Graph Sparsification, \texttt{One-Shot: max-Tree(Leverage) + Leverage}}
\label{alg:practical_sparsification}
\begin{algorithmic}[1]
\REQUIRE Graph $G=(V,E,w)$ with $|V|=n$, $|E|=m$; sparsity parameter $f$
\ENSURE Sparsified graph $H$

\STATE Compute effective resistance estimates (robust connectivities) $\hat{R}_e$ for all $e \in E$ \hspace{3em}// {Lines 1-11 from \cref{alg:spectral_sparsifier}, $\epsilon$ is only used after line 11 and does not have to be passed.}
\STATE Set scores $s_e \gets w_e \hat{R}_e$

\STATE $T \gets$ maximum spanning tree of $G$ using scores $s_e$
\STATE $E^\prime \gets E \setminus T$
\STATE $Z \gets \sum_{e \in E^\prime} s_e$

\FOR{each $e \in E^\prime$}
  \STATE $p_e \gets \min\!\left(1,\; \frac{(f\,m - n + 1)\,s_e}{Z}\right)$
  \STATE Sample $e$ with probability $p_e$. If sampled set $w_e \gets w_e/p_e$
\ENDFOR

\STATE \textbf{return} $H=\left(V,\; T \cup \widehat{E}\right)$
\end{algorithmic}
\end{algorithm}

{
\begin{remark}
    Given a graph $G=(V,E,w)$, let $\hat{R}_e$ be the effective resistance estimates as in \cref{thm:er_estimates}. If $\frac{fm - n + 1}{\sum_{e \in E'} w_e\hat{R}_e} > c_1\cdot \log n$ for an appropriate constant $c_1$, then \cref{alg:practical_sparsification} produces a $1/2$-approximate spectral sparsifier of $G$.
\end{remark}
\begin{proof}
    From \cref{thm:er_sampling} we know that if we oversample the edges of the graph with respect to their leverage scores $w_e R_e$ by a factor of at least $c \log n$, then we produce a valid spectral sparsifier. First, any edge that is in the sparsifier because it was in the maximum weight spanning tree (i.e. $T$ in \cref{alg:practical_sparsification}) is clearly oversampled. Second, we get robust connectivity estimates $\hat{R}_e$ from \cref{thm:er_estimates} (i.e. line 1 in \cref{alg:practical_sparsification}). If $\frac{fm - n + 1}{\sum_{e \in E'} w_e\hat{R}_e} > c_1\cdot \log n$ for an appropriate constant $c_1$, then, since $\hat{R}_e\geq R_e$, we oversample the actual leverage scores $R_e$ by $c_1\cdot \log n$. If we set $c_1 = c$, then we sample each edge with a probability larger than $c\log n \cdot w_eR_e$ and by invoking \cref{thm:er_sampling}, we achieve the desired guarantee.
\end{proof}
}

The work and depth of \cref{alg:practical_sparsification} are dominated by the computation of effective resistance estimates and consequently given by \cref{thm:main_algorithm}. \Cref{alg:practical_sparsification} will oftentimes use fewer edges than allowed by the budget. Let $E_{final}$ denote the number of edges in the final sparsifier and $T$ the computed maximum spanning tree. Then, we get 
\begin{align*}
        \E{X_{\text{final}}} &= n - 1 + \E{\sum_{e\in E\setminus T} p_e}\\
        &\leq n - 1 + \E{\sum_{e\in E\setminus T} \frac{(f\,m - n+1)\,s_e}{\sum_{e'\in E\setminus T} s_{e'}}}\\
        &= n-1 + fm - n + 1 = fm.
\end{align*}
This is mostly true for $f$ relatively large, in which case most sparsifiers work well.

\paragraph{A note on \cref{alg:practical_sparsification}.}
\Cref{alg:practical_sparsification} is a budgeted, one-shot variant of leverage-score sampling. Rather than invoking the oversampling lemma with an $\epsilon^{-2}\log n$ factor, we normalise the sampling probabilities to target a specific edge budget $f\cdot m$. We include a spanning-tree backbone $T$, chosen as the maximum spanning tree with respect to the same scores $s_e=w_e\hat{R}_e$. This guarantees connectivity of the output sparsifier. The overall design is in the spirit of backbone-based sparsification schemes (e.g. \cite{koutis2014,koutis2013fasterspectralsparsificationnumerical}), which retain a structured subgraph and then sample the remaining edges to meet a target density. Variants of this template could replace $T$ by other sparse backbones (such as spanners or low-stretch trees). In this work, we focus on the maximum-score tree since it is simple to compute and is known to be a good low-stretch tree in practice.

\paragraph{Baselines.}
We compare against the following baselines.
\begin{enumerate}[label=(\roman*)]
    \item \textbf{\texttt{Uniform Sampling}.} Independently sample each edge with a fixed probability chosen so that the expected number of retained edges matches $f\cdot m$, and reweight sampled edges by $w_e/p$.
    \item \textbf{\texttt{Leverage-Score Sampling}.} Independently sample each edge proportional to $s_e=w_e\hat{R}_e$, with probabilities normalized to match the expected budget $f\cdot m$, but without adding a backbone tree.
    \item \textbf{\texttt{One-Shot: max-Tree(weight) + Uniform}.} We run the same pipeline as \cref{alg:practical_sparsification}, but construct the backbone tree using the original weights $w_e$ (maximum spanning tree by $w_e$) and sample non-tree edges uniformly at random (with the uniform rate chosen to match the same expected edge budget).
    \item \textbf{Iterative variants.} We also consider an iterative version in which we apply four rounds of sparsification. We budget each round so that the final expected density matches the desired budget. We report both \texttt{Iterative: max-Tree(Leverage) + Leverage} and \texttt{Iterative: max-Tree(weight) + Uniform}.
\end{enumerate}
We see the algorithm \texttt{Iterative: max-Tree(weight) + Uniform} as a simplified version of the algorithm described by Koutis \cite{koutis2014}. Sampling a polylogarithmic number of spanners instead of a low-stretch tree would lead to additional density, which is why we opt for the max-weight tree. The rest of the algorithm is equivalent.

\subsection{Experiments}\label{sec:experiment_description}
The arguably most important application of spectral sparsifiers is as preconditioners to find solutions of partial differential equations. To that end, we construct two- and three-dimensional grid graphs with $n$ nodes, and assign edge weights in a checkerboard pattern with patchsizes of $4\times4$ and $4\times 4\times 4$ respectively. These graphs typically appear in diffusion problems with a jumping diffusion coefficient and are known to be among the more difficult diffusion problems \cite{gao2023robustpracticalsolutionlaplacian,amsel2026linearsystemseigenvalueproblems}. We set the smallest weight to be $1$, and the largest is $100{,}000$. Next, we eliminate a fraction $k$ of randomly chosen nodes using standard Cholesky elimination. This produces a smaller graph with $(1-k)n$ vertices that is much denser due to fill-in. In our experiments, we set $k=0.5$ for the 2d grids and $k=0.3$ for the 3d grids. We denote this dense graph by $G$. We then compute a sparsified version $\tilde{G}$ with $s \cdot m$ edges as the edge budget. We set $s =0.25$.

We measure two indicators for how well the sparsified graph approximates the original graph. When feasible, we report the condition number $\frac{\lambda(\tilde{G}^{-1}G)_{\text{max}} }{\lambda(\tilde{G}^{-1}G)_{\text{min}}} $ directly, which can be used to bound the convergence of iterative methods such as the Conjugate Gradient method. Since computing the condition number exactly takes $O(n^3)$ time, it is only feasible for relatively small graphs. We compute it for graphs where $n\leq 2000$. To measure how well our sparsifier approximates large graphs, we build an IC(0) preconditioner from $\tilde{G}$ and record how many PCG iterations are required to solve a linear system on $G$.

{Because all preconditioners are built to the same edge budget, their per-iteration solve cost is essentially identical, so the PCG iteration count serves as a hardware-independent proxy for total solve time. We confirm this and report end-to-end wall-clock timings in \cref{app:timing}.}

When measuring the PCG iterations, we run two types of experiments. First, we vary the number of nodes $n$, while keeping all other parameters fixed. This allows us to study the convergence asymptotics of the different preconditioners across different graph sizes. Second, we vary the desired sparsity of the preconditioner while keeping all other variables fixed. While it is clear that a sparser preconditioner has to decrease in performance, it allows us to study how well each preconditioning method utilises the edge budget it is given.

In \cref{app:other_cond_number_experiments}, we also compute the condition numbers for different $k$.

\subsection{Technical Specification}
All experiments were implemented in C++ using the Graph Based Benchmark Suite (GBBS)~\cite{dhulipala2018theoretically, dhulipala20gbbs} and executed in parallel on an ASUS ROG Strix laptop running Ubuntu 24.04.3 LTS. Code was compiled with \texttt{g++} 13.3.0.

\subsection{Results}
In this section, we present the practical results and benchmarks.

\subsubsection{Condition Numbers for Small Graphs}
In \cref{tab:cond_2d,tab:cond_3d}, we present the exact condition number of $\tilde{G}^{-1}G$. For the largest instance of the 2d grid, we see that our leverage score-based algorithm is approximately $5\times$ smaller than the best alternative baseline. Furthermore, we see that the condition number does not degenerate when increasing the size of the graph. Interestingly, we see that the iterative method of our proposed algorithm seems to be more stable for small graphs. The methods \texttt{Uniform Sampling} and \texttt{Leverage-Score Sampling} have an unbounded condition number even for small problem instances. This may happen when either sampling the graph disconnects it, in which case the smallest and second smallest eigenvalues are both zero or when the eigenvalues become so small that there appear numerical instabilities when computing them.

For the 3d grid, we see a similar behaviour, where our proposed method outperforms the baselines by about a factor of $4\times$, with the iterative version being slightly better than the one-shot version.
\begin{table}[t]
\centering
\caption{Condition numbers for the 2D dataset as a function of method and $n$. Entries are the median condition number over successful runs. Column abbreviations: Unif.\ = \texttt{Uniform Sampling}, Lev.\ = \texttt{Leverage-Score Sampling}, Tr-w+U = \texttt{One-Shot: max-Tree(weight) + Uniform}, Tr-lev+L = \texttt{One-Shot: max-Tree(Leverage) + Leverage}, It-U = \texttt{Iterative max-Tree(weight) + Uniform}, It-L = \texttt{Iterative max-Tree(Leverage) + Leverage}. $^{*}$ denotes number of failed runs; an em dash (—) denotes that all runs failed.}
\label{tab:cond_2d}
\begin{tabular}{rllllll}
\hline
$n$ & Unif. & Lev. & Tr-w+U & Tr-lev+L & It-U & It-L \\
\hline
50   & — & 16.3$^{**}$ & 51.4 & 23.8 & 65.4  & \textbf{16.3} \\
200  & — & $2.5\times 10^{15}$$^{*}$ & 48.7 & 18.9 & 71.2  & \textbf{13.9} \\
450  & — & 15.0$^{**}$ & 56.9 & 75.5 & 76.3  & \textbf{11.1} \\
800  & — & 17.8 & 60.4 & \textbf{6.8} & 77.7  & 8.4 \\
1250 & — & 20.6$^{*}$ & 61.0 & \textbf{11.5} & 126.9 & 11.8 \\
\hline
\end{tabular}
\end{table}

\begin{table}[t]
\centering
\caption{Condition numbers for the 3D dataset as a function of method and $n$. Entries are the median condition number over successful runs. Column abbreviations: Unif.\ = \texttt{Uniform Sampling}, Lev.\ = \texttt{Leverage-Score Sampling}, Tr-w+U = \texttt{One-Shot: max-Tree(weight) + Uniform}, Tr-lev+L = \texttt{One-Shot: max-Tree(Leverage) + Leverage}, It-U = \texttt{Iterative max-Tree(weight) + Uniform}, It-L = \texttt{Iterative max-Tree(Leverage) + Leverage}. $^{*}$ denotes number of failed runs; an em dash (—) denotes that all runs failed.}
\label{tab:cond_3d}
\begin{tabular}{rllllll}
\hline
$n$ & Unif. & Lev. & Tr-w+U & Tr-lev+L & It-U & It-L \\
\hline
511 & — & 30.32 & 36.35 & 13.66 & 83.05 & \textbf{10.96} \\
700 & — & 15.38$^{*}$ & 45.05 & \textbf{5.14} & 66.33 & 6.76 \\
932 & — & 59.11 & 52.95 & 17.07 & 61.00 & \textbf{10.72} \\
1210 & — & 22.21 & 51.46 & 8.12 & 50.46 & \textbf{5.43} \\
1538 & — & 25.64 & 46.64 & 10.27 & 86.09 & \textbf{8.48} \\
\hline
\end{tabular}
\end{table}

\subsubsection{Results for Large Graphs}
\Cref{fig:scaling_by_n} reports the number of PCG iterations as a function of the number of edges removed from the initial graph. The leverage score–based sparsifiers consistently outperform the weight-based baselines by a substantial margin. Moreover, the proposed sparsifier scales better with respect to the number of nodes than the weight-based alternatives. These results indicate that leverage score–based sparsification becomes increasingly relevant for larger graphs, which is exactly the regime where leverage score computation is also more demanding.

\begin{figure}[htbp]
    \centering
    \includegraphics[width=0.5\textwidth]{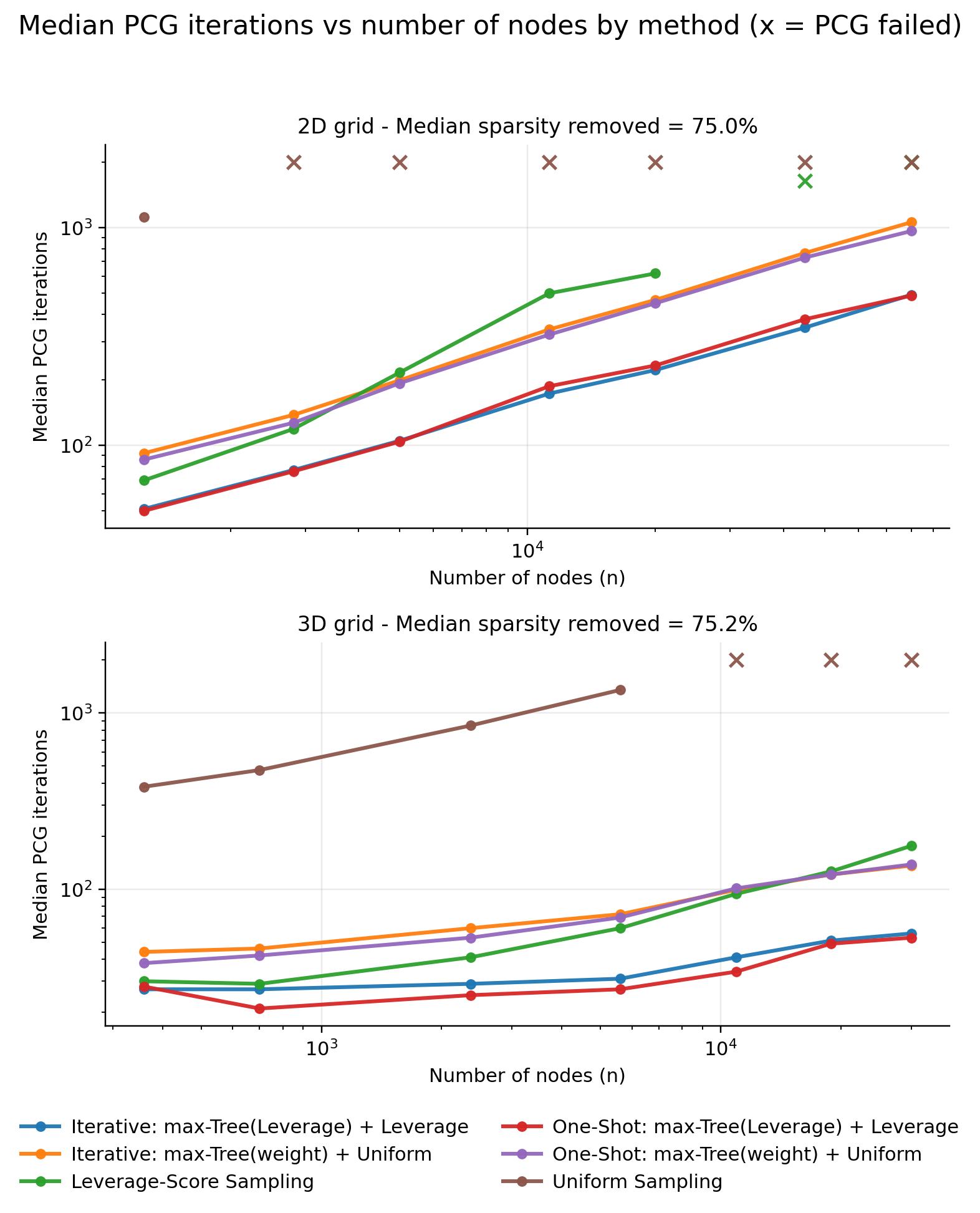}
    \caption{The number of PCG iterations as a function of the size of the graph. We first sparsify the original graph and then use IC(0) as a preconditioner. The budget for the edges is $25\%$ of the initial, unsparsified number of edges, and we report the percentage that was actually removed. The result is the median number of PCG iterations across 3 runs. Failed runs (more than 2000 PCG iterations) are marked with a $\times$ in the corresponding color.}
    \label{fig:scaling_by_n}
    \Description{Median PCG iterations versus number of nodes for each sparsification method.}
\end{figure}

\Cref{fig:pcg_by_sparsity} reports the number of PCG iterations as a function of the number of edges removed from the original graph $G$. For each graph family, we evaluate the algorithm on two graph sizes. Across all experiments, we observe qualitatively similar behaviour: the leverage-score-based sparsifier consistently outperforms the baselines by a substantial margin. At moderate sparsification levels, the leverage-score-based method also exhibits better scaling as sparsity increases. However, its performance worsens once the sparsified graph is required to be extremely sparse ($10\%$ of initial edges). A plausible explanation is that, in this regime, the variance introduced by independent sampling becomes more pronounced.

\begin{figure}[htbp]
    \centering
    \includegraphics[width=0.5\textwidth]{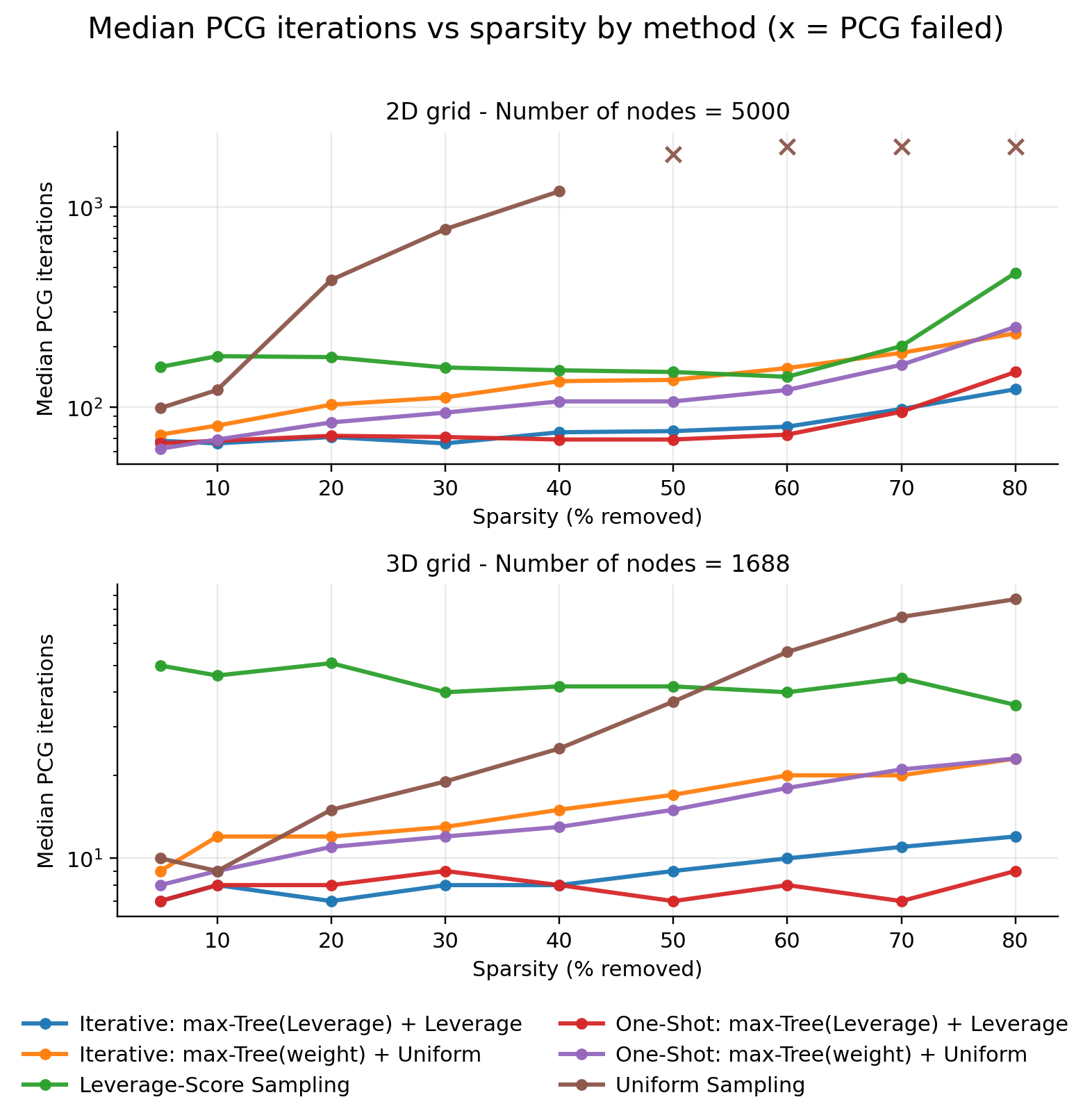}
    \caption{The number of PCG iterations as a function of the sparsity of the graph. We first sparsify the original graph and then use IC(0) as a preconditioner. We report the median number of PCG iterations out of three runs. Failed runs (more than 2000 PCG iterations) are marked with a $\times$ in the corresponding color.}
    \label{fig:pcg_by_sparsity}
    \Description{Median PCG iterations versus number of nodes for each sparsification method.}
\end{figure}

\paragraph{Remark.} Iterative sparsification with the proposed method behaves very similarly to one-shot sparsification. Because the tree is recomputed from the leverage scores, the trees selected in consecutive rounds can be nearly identical and effective resistances are preserved throughout the iterative sparsification.

\section{Improving parallel approximate Cholesky}
In \cite{kyng2016approximategaussianeliminationlaplacians,gao2023robustpracticalsolutionlaplacian} the authors present algorithms that take a Laplacian matrix $L$ and produce a factorization $\mathcal{L}\mathcal{L}^\top$ that spectrally approximates $L$ and which can be used to solve linear systems fast. They essentially show that when running the standard Cholesky algorithm on the Laplacian matrix $L$, it suffices to insert a sparsified Schur complement at each step to achieve a good factorization. The algorithms themselves differ only in how the sparsified Schur complement is computed. To be more formal, let $L^{(0)}$ be the initial input Laplacian, where the vertices are randomly permuted. Then, the Cholesky algorithm, at step $i$, computes the factor $\mathcal{L}$ as 
\begin{equation*}
    \mathcal{L}[:, i] = \frac{L^{(i)}[:,i]}{\sqrt{L^{(i)}[i,i]}} \quad \text{ and } \quad L^{(i+1)} = L^{(i)} - \frac{L^{(i)}[:,i]\cdot L^{(i)}[i,:]}{L^{(i)}[i,i]}
\end{equation*}
where $L[:,i]$ and $L[i,:]$ are the $i$-th row and column of $L$. It can be shown that not only $L^{(i)}$ is a Laplacian for any $i$, but also that the correction term can be decomposed into two Laplacians as follows 
\begin{equation*}
    \frac{L^{(i)}[:,i]\cdot L^{(i)}[i,:]}{L^{(i)}[i,i]} = \textsc{STAR($L^{(i)},v_i$)} - \textsc{CLIQUE($L^{(i)},v_i$)}.
\end{equation*}
The \textsc{CLIQUE($L^{(i)},v_i$)} Laplacian is essentially a clique on the neighbors of the vertex that is eliminated at stage $i$. The algorithms of \cite{kyng2016approximategaussianeliminationlaplacians,gao2023robustpracticalsolutionlaplacian} now show that it suffices to replace the \textsc{CLIQUE($L^{(i)},v_i$)} Laplacian by a sparsified version of it.

In \cite{BK24}, it was then shown that these algorithms can be parallelized, with polylogarithmic depth, through a framework that, in addition to sparsifying each individual Schur complement, sparsifies the remaining graph $L^{(i)}$ at only $O(\log n)$ specific steps of the elimination process. While any black-box sparsification algorithm suffices to achieve polylogarithmic depth, the authors point out that the overall algorithm's work and depth are bottlenecked by the work and depth of the sparsification algorithm. In fact, to achieve a factorization that is $\epsilon$-approximate, the algorithm needs to compute an $\frac{\epsilon}{\log n}$-approximate spectral sparsifier $O(\log n)$ times. When the sparsification algorithm proposed by \cite{koutis2014} is used, then the depth and work of the sparsification algorithm scale with $\tilde{\epsilon}^{-2}$, where $\tilde{\epsilon}$ is the desired accuracy of the spectral sparsifier, and thus the total depth and work of the approximate Cholesky algorithm has an additional $\log^2n$ factor in the depth and work. Our proposed algorithm in \cref{thm:main_algorithm} has, in addition to lower work and depth overall, no dependence on $\epsilon$, and, when used as a black box, immediately improves the depth and work of the parallelization framework. In particular, it shifts the depth bottleneck from the computational depth of the sparsification routine to the sparsity of the resulting sparsifier.

\section{Conclusion}
We have presented a solver-free, parallel spectral sparsification algorithm that significantly reduces both work and depth compared to prior methods. Extremely simple in design, our approach relies only on sub-sampling and the computation of low-diameter decompositions. To demonstrate its practical impact, we implemented a heuristic based on the algorithm and observed strong performance across a set of graphs and sparsity regimes.

\section{Future Work}
We believe that the LDD-based spectral sparsification algorithm we presented could potentially also be turned into a dynamic algorithm by using dynamic versions of low-diameter decompositions, which exist \cite{10.1145/3313276.3316381}. Consequently, it would be interesting to see if one can extend the algorithm to work in a parallel batch-dynamic setting.

We also believe that the key application of graph sparsification algorithms is for preconditioning Laplacian linear systems. In practice, Laplacian system solvers are highly optimized programs that run in time $O(m\log n)$ on huge graphs. As such, using spectral sparsification as a subroutine for these solvers would imply that the solvers need not only be theoretically fast, but also scale well in practice. For undirected Laplacian linear systems, we now have very simple spectral sparsification algorithms. However, more recent work has extended Laplacian solvers to directed, Eulerian graphs. It would be interesting to see to what extent sparsifiers could be used to speed up directed Laplacian solvers and whether simple enough sparsifiers for directed Laplacian linear systems can exist to speed up solving linear systems in practice.

\newpage
\bibliography{refs}
\appendix
\section{Appendix}\label{app:other_cond_number_experiments}
\subsection{Condition numbers for slightly different graphs}
In this section, we compare how the computed condition numbers behave if we vary the initial density of the sparsified graph.
\subsubsection{Denser Graphs}
\Cref{tab:cond_2d_different_density,tab:cond_3d_different_density} show the condition numbers for the same experiment as run in \cref{sec:experiment_description} but with $k=0.6$ (vs. $k=0.5$ before) for the 2d grids and $k=0.3$ (vs. $k=0.4$ before). This leads to an overall denser graph. 

We observe that the proposed sparsifier outperforms the baselines by a wide margin. Increasing the density of the original graph also increases the edge budget $mf$ relative to the number of nodes in the graph $n$, since we increase $m$ while keeping $n$ fixed. From \cref{fig:pcg_by_sparsity} we already know that our proposed sparsifier better uses the budget available. The same can be observed in this instance.

\begin{table}[t]
\centering
\caption{Condition numbers for the 2D dataset as a function of method and $n$. Entries are the median condition number over successful runs. Column abbreviations: Unif.\ = \texttt{Uniform Sampling}, Lev.\ = \texttt{Leverage-Score Sampling}, Tr-w+U = \texttt{One-Shot: max-Tree(weight) + Uniform}, Tr-lev+L = \texttt{One-Shot: max-Tree(Leverage) + Leverage}, It-U = \texttt{Iterative max-Tree(weight) + Uniform}, It-L = \texttt{Iterative max-Tree(Leverage) + Leverage}. $^{*}$ denotes number of failed runs; an em dash (—) denotes that all runs failed.}
\label{tab:cond_2d_different_density}
\begin{tabular}{rllllll}
\hline
$n$ & Unif. & Lev. & Tr-w+U & Tr-lev+L & It-U & It-L \\
\hline
40   & — & $5.5\times 10^{16}$$^{*}$ & 21.7 & 24.3 & 39.4 & \textbf{12.5} \\
160  & — & 9.1  & 15.2 & \textbf{3.5} & 21.5 & 6.8 \\
360  & — & 9.7  & 25.9 & \textbf{3.3} & 42.1 & 6.2 \\
640  & — & 19.8 & 26.2 & \textbf{2.6} & 33.6 & 5.0 \\
1000 & — & 45.7 & 30.6 & \textbf{2.6} & 47.2 & 5.2 \\
\hline
\end{tabular}
\end{table}

\begin{table}[t]
\centering
\caption{Condition numbers for the 3D dataset as a function of method and $n$. Entries are the median condition number over successful runs. Column abbreviations: Unif.\ = \texttt{Uniform Sampling}, Lev.\ = \texttt{Leverage-Score Sampling}, Tr-w+U = \texttt{One-Shot: max-Tree(weight) + Uniform}, Tr-lev+L = \texttt{One-Shot: max-Tree(Leverage) + Leverage}, It-U = \texttt{Iterative max-Tree(weight) + Uniform}, It-L = \texttt{Iterative max-Tree(Leverage) + Leverage}. $^{*}$ denotes number of failed runs; an em dash (—) denotes that all runs failed.}
\label{tab:cond_3d_different_density}
\begin{tabular}{rllllll}
\hline
$n$ & Unif. & Lev. & Tr-w+U & Tr-lev+L & It-U & It-L \\
\hline
438  & — & 14.3 & 20.7 & \textbf{2.3} & 37.0 & 5.1 \\
600  & — & 29.3 & 21.3 & \textbf{2.1} & 24.3 & 4.3 \\
799  & — & 22.1 & 21.7 & \textbf{2.4} & 32.6 & 3.8 \\
1037 & — & 51.4 & 21.6 & \textbf{2.2} & 21.3 & 4.0 \\
1319 & — & 54.5 & 37.5 & \textbf{2.2} & 42.9 & 4.1 \\
\hline
\end{tabular}
\end{table}

\subsection{End-to-end running times}
\label{app:timing}
\Cref{tab:timing} reports wall-clock timings on the 2D grid, measured on 24 threads. We focus on the two one-shot methods. The data support using PCG iteration count as the headline metric in \cref{sec:experiment_description}. First, once the edge budget is fixed the sparsifiers have essentially equal edge counts ($\approx 25\%$ of the original) and the per-iteration solve cost is identical across methods to within $1\%$, so iteration count is a hardware-independent proxy for solve time. Second, although estimating leverage scores makes the leverage-based construction more expensive than the weight-based baseline, the reduction in iterations more than compensates: the leverage-based method is faster end-to-end at both sizes (e.g.\ $268$\,s vs.\ $342$\,s at $n=720{,}000$). Finally, our construction uses only LDDs, MSTs, array sums, and edge subsampling, all of which admit scalable parallel implementations~\cite{dhulipala2018theoretically}.

\begin{table*}[t]
\centering
\small
\caption{Spectral sparsifier scaling results on the 2D grid, measured on 24 threads. Rows show the median over 3 runs. Only the one-shot methods are included; the abbreviations Tr-w+U and Tr-lev+L are as in Table~\ref{tab:cond_2d}. Here $m_{\text{orig}}$ and $m_{\text{spars}}$ are the edge counts before and after sparsification, $t_{\text{sparsify}}$ is the sparsifier construction time, $t_{\text{solve}}$ the total PCG solve time, and $t_{\text{solve}}/\text{iter}$ the per-iteration solve cost.}
\label{tab:timing}
\begin{tabular}{lrrrrrrrr}
\hline
Method & $n$ & $m_{\text{orig}}$ & $m_{\text{spars}}$ & $m_{\text{spars}}/m_{\text{orig}}$ & Iters & $t_{\text{sparsify}}$ (s) & $t_{\text{solve}}$ (s) & $t_{\text{solve}}/\text{iter}$ (ms) \\
\hline
Tr-w+U   & 720,000   & 20,976,346 & 5,243,755 & 25.0\% & 4,668 & 1.28  & 340.2 & 72.9  \\
Tr-lev+L & 720,000   & 20,976,346 & 5,305,784 & 25.3\% & 3,252 & 28.80 & 239.6 & 73.7  \\
Tr-w+U   & 1,125,000 & 33,024,333 & 8,255,914 & 25.0\% & 5,264 & 1.92  & 642.2 & 122.0 \\
Tr-lev+L & 1,125,000 & 33,024,333 & 8,348,455 & 25.3\% & 3,841 & 46.14 & 471.9 & 122.9 \\
\hline
\end{tabular}
\end{table*}

\end{document}